\documentclass[letterpaper,10pt,conference]{ieeeconf}

 \IEEEoverridecommandlockouts
\usepackage[T1]{fontenc}

\usepackage[table]{xcolor}
\usepackage{multirow}
\usepackage{graphics} 

\graphicspath{{./fig/}}

\usepackage{tikz}
\usetikzlibrary{positioning}
\usetikzlibrary{arrows}
\usepackage{epstopdf}
\usepackage{epsfig} 
\usepackage{subcaption}
\usepackage{psfrag}
\usepackage[width=.95\linewidth]{caption}
\usepackage{lipsum}
\usepackage{amsmath} 
\usepackage{amssymb}  
\usepackage{mathbbol}
\usepackage{cite}
\usepackage{wrapfig}
\usepackage{mathrsfs}
\usepackage{url}
\usepackage{amsmath,amsfonts, amssymb,color}
\usepackage[font=footnotesize]{caption}
\usepackage[pdfpagelabels,pdfpagemode=None,breaklinks=true]{hyperref} \hypersetup{
     colorlinks   = true,
     citecolor    = blue
}
\usetikzlibrary{shapes,arrows,calc}
\usepackage{marvosym}

\newtheorem{lemma}{Lemma}
\newtheorem{remark}{Remark}
\newtheorem{proposition}{Proposition}

\makeatletter
\def\@opargbegintheorem#1#2#3{\trivlist
   \item[]{\it #1\ #2\ (#3)} \itshape}
\makeatother

\newcommand{\ignore}[1]{}

\DeclareMathAlphabet{\mathcal}{OMS}{cmsy}{m}{n}      

\newcommand{\setdef}[2]{\{#1 \; | \; #2\}}

\newcommand{\SUt}{\mathtt{SU}}
\newcommand{\SPt}{\mathtt{SP}}
\newcommand{\IPt}{\mathtt{IP}}
\newcommand{\IUt}{\mathtt{IU}}

\newcommand{\oprocendsymbol}{\hbox{$\bullet$}}
\newcommand{\oprocend}{\relax\ifmmode\else\unskip\hfill\fi\oprocendsymbol}

\allowdisplaybreaks

\newcommand{\subscr}[2]{#1_{\textup{#2}}}

\renewcommand{\theenumi}{(\roman{enumi}}

\def \mr{\mathrm}

\begin{document}

\title{Epidemic Propagation under Evolutionary Behavioral Dynamics: \\ Stability and Bifurcation Analysis
	\thanks{This work was supported in part by ARO grant W911NF-18-1-0325.}
}
\author{Abhisek Satapathi, Narendra Kumar Dhar, Ashish R. Hota, and Vaibhav Srivastava
	\thanks{A. Satapathi and A. R. Hota are with the Department of Electrical Engineering, IIT Kharagpur, Kharagpur, West Bengal, India, 721302. {\small Email: \tt{abhisek.ee@iitkgp.ac.in,ahota@ee.iitkgp.ac.in}}}
	\thanks{N. K. Dhar and V. Srivastava are with the Department of Electrical and Computer Engineering, Michigan State University, East Lansing, MI 48824. {\small Email: \tt\{dharnare, vaibhav\}@msu.edu}}
}

\maketitle

\begin{abstract}
We consider the class of SIS epidemic models in which  a large population of individuals chooses whether to adopt protection or to remain unprotected as the epidemic evolves. For a susceptible individual, adopting protection reduces the probability of becoming infected but it comes with a cost that is weighed with the instantaneous risk of becoming infected. An infected individual adopting protection transmits a new infection with a smaller probability compared to an unprotected infected individual.  We focus on the replicator evolutionary dynamics  to model the evolution of protection decisions by susceptible and infected subpopulations. We completely characterize the existence and local stability of the equilibria of the resulting coupled epidemic and replicator dynamics. We further show how the stability of different equilibrium points gets exchanged as certain parameters change. Finally, we investigate the system behavior under timescale separation between the epidemic and the evolutionary dynamics.
\end{abstract}


\IEEEpeerreviewmaketitle


\section{Introduction}

Infectious diseases or epidemics spread through human society via social interactions among infected and healthy individuals. There is a rich body of work devoted to developing mathematical models that capture the dynamic evolution of such epidemics \cite{mei2017dynamics,nowzari2016analysis}. In the well-studied susceptible-infected-susceptible (SIS) epidemic model, an individual belongs to one of two infectious states: susceptible (healthy) and infected. A susceptible individual becomes infected with a certain probability termed infection rate if it comes in contact with an infected individual, while an infected individual recovers with a certain recovery rate. 

Individuals adopt protective measures, such as wearing masks or becoming vaccinated, to avoid becoming infected, often in a strategic and decentralized manner. Consequently, the interplay between epidemic dynamics and human decision-making has been studied in the framework of game theory \cite{chang2020game}; both static or single-shot games to model vaccination decisions \cite{hota2019game,trajanovski2015decentralized} and dynamic games that model evolution of protective decisions \cite{eksin2016disease,huang2019differential,hota2020impacts} have been analyzed for SIS epidemics in recent past. Most of these works explore structural properties of the equilibrium strategies and epidemic evolution under those strategies. 

However, the problem of a large group of strategic individuals finding or learning equilibrium strategies is quite challenging \cite{daskalakis2009complexity,sandholm2010population}. In settings with a large population of agents, {\it evolutionary dynamics} have been proposed and their convergence behavior to the set of equilibria is analyzed \cite{hofbauer2003evolutionary,sandholm2010population}. This class of learning dynamics involves repeated play or interaction of a static game with the same set of payoff functions. Evolutionary learning with dynamic payoff functions has also recently been explored by exploiting the concept of passivity \cite{arcak2020dissipativity,park2019payoff}. In contrast with the above settings, epidemic games exhibit a dynamically changing proportion of healthy and infected individuals due to epidemic dynamics. A recent work \cite{elokdadynamic_cdc_arxiv} investigated this dynamic evolution of infection dynamics and strategic decisions for a class of susceptible-asymptomatic-infected-recovered (SAIR) epidemic setting in the framework of dynamic population games \cite{elokda2021dynamic}.

Timescale separation has been used to study the interaction  of both evolutionary as well as epidemic dynamics with other dynamical systems. For example, 
slow-fast coupled environmental and evolutionary game dynamics were first studied in \cite{weitz2016oscillating}. A recent paper \cite{gong2021different} also studies slow-fast dynamics in a setting related to environment and ecology. In the context of epidemics, a recent work \cite{al2021long} explores slow-fast dynamics in the context of SIR epidemics but does not consider a behavioral response to epidemics or any game-theoretic considerations. SIS dynamics with risk-tolerator or risk-averter individuals that respond  differently as a function of the leaky-integrated infectious population are studied in \cite{zhou2020active};  however, game-theoretic dynamics of such behavioral responses have not been investigated. Coupled evolution of disease and behavior has recently been investigated in \cite{martins2022epidemic,khazaei2021disease} for SIRS and SEIR epidemic models, respectively. 

In this paper, we build upon the above line of work and consider the class of SIS epidemic models where a large population of individuals chooses whether to adopt protection or to remain unprotected as the epidemic evolves. For a susceptible individual, being protected reduces the probability of becoming infected. Similarly, an infected individual who adopts protection transmits or causes a new infection with a smaller probability compared to an infected individual who is unprotected. For a susceptible individual, adopting protection comes with a cost (such as discomfort due to wearing masks or adhering to social distancing protocols) which is weighed with the instantaneous risk of becoming infected; the latter depends on the current epidemic prevalence and the proportion of individuals adopting protection. 

In order to lay the foundations for a rigorous investigation of coupled epidemic and behavioral dynamics, we focus on the {\it replicator dynamics} \cite{hofbauer2003evolutionary,sandholm2010population} to model the evolution of protection decisions by susceptible and infected subpopulations. We completely characterize the equilibria (existence and local stability) of the coupled epidemic and replicator dynamics and show how the stability of different equilibrium points gets exchanged as certain parameters change. 

We further explore the behavior of the coupled dynamics under timescale separation leading to a slow-fast dynamical system \cite{berglund2006noise}. We focus on two timescale separation scenarios. First, 
we analyze the  case in which the replicator dynamics are faster and the behavior converges to equilibrium strategies quickly leading to a variable structure system as observed in prior works (such as \cite{hota2020impacts}). We then consider the case where the epidemic dynamics are faster than the replicator dynamics. We numerically illustrate the dynamic bifurcation behavior and the bifurcation delay phenomenon. Our results show that the timescale separation of disease and behavioral dynamics plays a significant role in disease evolution, and may even lead to inadequate and counterproductive epidemic control measures if not properly accounted for. 
\section{Coupled SIS Epidemic and Evolutionary Behavioral Model}
\label{section:model}

In this section, we formally introduce the coupled evolution of the SIS epidemic and protection adoption  behavior in a homogeneous large-population setting. Let the proportion of susceptible and infected individuals be $s(t)$ and $y(t)$, respectively. Both $s(t), y(t) \in [0,1]$ with $s(t) + y(t) = 1$ for all $t \geq 0$. Individuals choose whether to adopt protection against the epidemic or not; these actions are denoted by $\mathtt{P}$ and $\mathtt{U}$, respectively. Consequently, the {\it population state} at time $t$ is defined as $x(t) := [x_{\SUt}(t) \quad x_{\SPt}(t) \quad x_{\IUt}(t) \quad x_{\IPt}(t)]^\intercal \in [0,1]^4$, where $x_{\SUt}$ denotes the proportion of individuals who are susceptible and choose to remain unprotected, $x_{\IPt}$ denotes the proportion of infected individuals who adopt protection, and so on. At time $t$, we have $x_{\SUt}(t) + x_{\SPt}(t) = s(t)$, $x_{\IUt}(t) + x_{\IPt}(t) = y(t)$, and $\mathbb{1}^\intercal x(t) = 1$ where $\mathbb{1}$ is a vector of appropriate dimension with all entries $1$. 

Individuals choose their action to maximize their payoffs which depends on the population state $x(t)$ (which includes information regarding infection prevalence since $x_{\IUt}(t) + x_{\IPt}(t) = y(t)$). For an infected individual, there is no further risk of infection, and as a result, we define its payoff to be constant parameters given by $-c_{\IUt}$ if it remains unprotected and $- c_{\IPt}$ if it adopts protection, respectively. In particular, $c_{\IUt} > 0$ captures the cost of breaking isolation/quarantine protocols for an infected individual or could represent worsening of the disease if an infected individual continues to be irresponsible. Thus, we assume $c_{\IUt} > c_{\IPt} \geq 0$.

A susceptible individual trades off the cost of adopting protection, denoted by $c_P > 0$, and the expected cost of becoming infected. The latter is the product of the loss upon infection $L > 0$ and the instantaneous probability of becoming infected which depends on its action as well as the population state. Specifically, let $\beta_u$ and $\beta_p$ denote the probabilities of an infected individual causing a new infection if it is unprotected and protected, respectively. We assume that $\beta_u > \beta_p \geq 0$. Similarly, let a protected susceptible individual be $\alpha \in (0,1)$ times (less) likely to become infected compared to an unprotected susceptible individual. 

Thus, the payoff vector of individuals at population state $x$ is now defined as
\begin{equation}
F(x) := 
\begin{bmatrix}
F_{\SUt}(x) \\
F_{\SPt}(x) \\
F_{\IUt}(x) \\
F_{\IPt}(x)
\end{bmatrix}
=
\begin{bmatrix}
-L (\beta_u x_{\IUt} + \beta_p x_{\IPt}) \\
-c_P - L \alpha (\beta_u x_{\IUt} + \beta_p x_{\IPt}) \\
- c_{\IUt} \\
- c_{\IPt}
\end{bmatrix}
,
\end{equation}
where $F_{\SUt}(x)$ denote the payoff for an individual who is susceptible and unprotected at population state $x$, and so on. Susceptible individuals who adopt protection pay a cost $c_P$ but experience a reduced risk of becoming infected scaled by factor $\alpha$ as discussed above. 

The population state $x$ is dynamic and time-varying, both due to the epidemic dynamics as well as the evolutionary dynamics of strategically adopting protective behavior. Since the proportion of susceptible and infected individuals is time-varying, the framework of classical population games is not applicable. We now introduce the coupled epidemic and evolutionary dynamics. We first introduce some notation. Let $z_S(t) \in [0,1]$ denote the fraction of susceptible nodes who remain unprotected, i.e., $x_{\SUt}(t) = z_S(t) s(t)$ and $x_{\SPt}(t) = (1-z_S(t)) s(t)$. Similarly, let $z_I(t) \in [0,1]$ denote the fraction of infected nodes who remain unprotected. 

Due to the presence of both unprotected and protected individuals with different infection probabilities, the evolution of infected proportion is given by
\begin{align}
\dot{y}(t) & = (x_{\SUt}(t) + \alpha x_{\SPt}(t)) (\beta_u x_{\IUt}(t) + \beta_p x_{\IPt}(t)) - \gamma y(t) \nonumber
\\ & = \big[(1-y(t)) (z_S(t) + \alpha (1-z_S(t))) \nonumber
\\ & \qquad \qquad \cdot (\beta_u z_I(t) + \beta_p (1-z_I(t))) - \gamma \big] y(t) \nonumber
\\ & =: f_y(y(t),z_S(t),z_I(t)). \label{eq:sis_scalar_com}
\end{align}
The above dynamics are analogous to the conventional scalar SIS epidemic dynamics with effective infection rate $ \subscr{\beta}{eff} = (z_S(t) + \alpha (1-z_S(t))) (\beta_u z_I(t) + \beta_p (1-z_I(t)))$ which now depends on the efficacy of protection and the fractions that adopt protection.\footnote{In particular, if $\alpha = 1$ and $\beta_p = \beta_u = \beta$, i.e., protection is not effective, the effective infection rate is $\beta$.} 


\begin{table*}[ht]
	\centering
	\caption{Existence and stability of equilibria of the coupled dynamics.}
	\label{tab:eq_summary}
	{\renewcommand{\arraystretch}{1.5}
		\begin{tabular}{|c |c |c |c |c |c |}
			\hline
			\multirow{2}{*}{Epidemic} & \multirow{2}{*}{Endemic} & \multicolumn{4}{c|}{Equilibria} \\ \cline{3-6} 
			
			Parameters & Infection Level & $\mathbf{E1}: (0,1,0)$ & $\mathbf{E2}: (y^*_u,1,0)$ & $\mathbf{E3}: (y^*_{\mr{int}},z^*_{S,int},0)$ & $\mathbf{E4}: (y^*_p,0,0)$ \\ \hline \hline
			
			$\gamma > \beta_p$ & $-$ & \cellcolor{green!45} $\checkmark$, stable & $-$ & $-$ & $-$ \\ \hline
			
			\multirow{2}{*}{$\alpha \beta_p < \gamma < \beta_p$} & $y^*_u < y^*_{\mr{int}}$ & $\checkmark$, unstable & \cellcolor{green!45} $\checkmark$, stable  & $-$  & $-$ \\ \cline{2-6} 
			
			& $y^*_{\mr{int}} < y^*_u$ & $\checkmark$, unstable & $\checkmark$, unstable & \cellcolor{green!45} $\checkmark$, stable & $-$ \\ \hline
			
			\multirow{3}{*}{$\gamma < \alpha \beta_p$} & $y^*_{u} < y^*_{\mr{int}}$ & $\checkmark$, unstable & \cellcolor{green!45} $\checkmark$, stable &  $-$ & $\checkmark$, unstable \\ \cline{2-6} 
			
			& $y^*_p < y^*_{\mr{int}} < y^*_u$ & $\checkmark$, unstable & $\checkmark$, unstable  & \cellcolor{green!45} $\checkmark$, stable & $\checkmark$, unstable \\ \cline{2-6}
			
			& $y^*_{\mr{int}} < y^*_p$ & $\checkmark$, unstable & $\checkmark$, unstable & $-$ & \cellcolor{green!45} $\checkmark$, stable \\ \hline
			
	\end{tabular}}
\end{table*}

We now define the evolutionary dynamics to capture the evolution of unprotected fractions in both susceptible and infected subpopulations. We focus on the class of replicator dynamics~\cite{sandholm2010population} in this work and assume that susceptible nodes only replicate the strategies of other susceptible nodes (likewise for infected nodes). Consequently, for susceptible nodes, we obtain
\begin{align}
\dot{z}_{S}(t) & =  {z}_{S}(t)(1-{z}_{S}(t)) \big[ [F]_{\SUt} - [F]_{\SPt} \big] \nonumber
\\ & = {z}_{S}(t)(1-{z}_{S}(t))\big[ c_P - L(1-\alpha)(\beta_u z_{I}(t) \nonumber
\\ & \qquad \qquad + \beta_p (1-z_I(t))) y(t) \big] \nonumber
\\ & =: f_S(y(t),z_S(t),z_I(t)). \label{eq:main_zs}
\end{align}

Similarly, for infected nodes, we have
\begin{align}
\dot{z}_{I}(t) & = {z}_{I}(t) \big[ - c_{\IUt} + (c_{\IUt} {z}_{I}(t) + c_{\IPt} (1-{z}_{I}(t)))  \big] \nonumber
\\ & = {z}_{I}(t) (1-{z}_{I}(t)) (c_{\IPt}-c_{\IUt}) \nonumber
\\ & =: f_I(y(t),z_S(t),z_I(t)). \label{eq:main_zi}
\end{align}

Thus, equations \eqref{eq:sis_scalar_com}, \eqref{eq:main_zs}, and \eqref{eq:main_zi} characterize the coupled evolution of the epidemic and population states. 

\begin{lemma}[Invariant set]\label{lem:invariant-set}
For the coupled SIS epidemic and evolutionary behavior dynamics defined by \eqref{eq:sis_scalar_com}, \eqref{eq:main_zs} and \eqref{eq:main_zi}, the set $\setdef{(y, z_S, z_I)}{(y, z_S, z_I) \in [0,1]^3}$ is invariant. 
\end{lemma}
\begin{proof}
	For $y(t) \in \{0,1\}$, $\dot y$ is either zero or negative. Likewise, for $z_I, z_S \in \{0,1\}$, $\dot z_I = \dot z_S =0$. The result is immediate from Nagumo's theorem~\cite[Theorem 4.7]{blanchini2008set}.  
	\end{proof}
	
\section{Equilibrium Characterization and Bifurcation Analysis}\label{sec:stability-analysis}

In this section, we first examine the equilibrium points of the above coupled epidemic-replicator dynamics and their stability properties. We then examine the bifurcations associated with the changes in the stability of these equilibria. 

\subsection{Equilibrium Characterization and Stability Analysis}

First, we consider the evolution of $z_I$ in \eqref{eq:main_zi} which does not depend on $y$ and $z_S$. There are two stationary points $z_I = 0$ and $z_I = 1$, and it is easy to see that, for $c_{\IUt} > c_{\IPt}$, $z_I = 1$ is unstable and $z_I = 0$ is exponentially stable with basin of attraction $[0,1)$. It is also quite intuitive that if $c_{\IPt} < c_{\IUt}$, infected nodes prefer to use protection, and the strategy for infected nodes should converge to adopting protection. 

Thus, in the remainder of this paper, we only focus on equilibria with $z_I = 0$. We begin with introducing a few variables that will be used to define the equilibrium points. 
	\begin{multline*}
		y^*_u : = 1 - \frac{\gamma}{\beta_p}, \quad  y^*_{\mr{int}}:= \frac{c_P}{L(1-\alpha)\beta_p}, \quad
		y^*_p := 1 - \frac{\gamma}{\alpha \beta_p}, \\  \text{and } z^*_{S,\mr{int}} := \frac{1}{1-\alpha} \left[ \frac{\gamma}{\beta_p(1-y^*_{\mr{int}})} - \alpha \right].
	\end{multline*}
	
	We now define all possible equilibria of the coupled SIS epidemic and evolutionary behavior dynamics  (\ref{eq:sis_scalar_com}--\ref{eq:main_zi}) corresponding to $z_I=0$:
	\begin{multline*}
		\mathbf{E0}= (0,0,0), \quad 	\mathbf{E1}= (0,1,0),  \quad  \mathbf{E2}= (y^*_u,1,0), \\
		\mathbf{E3}= (y^*_{\mr{int}},z^*_{S,\mr{int}},0),  \text{ and}  \quad  \mathbf{E4}= (y_p^*,0,0). 
	\end{multline*}
	
At $\mathbf{E0}$ everyone adopts protection and there is no infection. At $\mathbf{E1}$, there is no infection, and susceptible individuals do not adopt protection. $ \mathbf{E2}$ is an endemic equilibrium, i.e., a fraction of the population is infected, at which  susceptible individuals continue to remain unprotected. $\mathbf{E3}$ is an endemic equilibrium  at which some  susceptible individuals adopt protection. Finally, $\mathbf{E4}$ is an endemic equilibrium  at which all  susceptible individuals adopt protection. The existence and stability of these equilibria are established below and summarized in Table \ref{tab:eq_summary}. 

\begin{proposition}[Equilibria and Stability]\label{prop:equilibria-and-stability}
For the equilibrium points  of the coupled SIS epidemic and evolutionary behavioral dynamics  (\ref{eq:sis_scalar_com}--\ref{eq:main_zi}) corresponding to $z_I=0$, the following statements hold: 
\begin{enumerate}
\item $\mathbf{E0}$ exists for all parameter regimes,  and is unstable; 
\item $\mathbf{E1}$ exists for all parameter regimes, is stable if $\beta_p < \gamma$, and is unstable, otherwise; 
\item $ \mathbf{E2}$ exists only when $\beta_p > \gamma$, is stable when $y^*_u < y^*_{\mr{int}}$, and is unstable otherwise; 
\item $\mathbf{E3}$ exists only when $ y^*_p< y^*_{\mr{int}} <  y^*_u$, and is stable; 
\item  $\mathbf{E4}$ exists only when $\gamma < \alpha \beta_p$, is stable when $y^*_p > y^*_{\mr{int}}$, and is unstable otherwise. 
\end{enumerate}
\end{proposition}

The proof is presented in the Appendix.

\ignore{
\begin{proof}

\end{proof}
}

\subsection{Bifurcation Analysis}\label{sec:bif-analysis}

We now numerically explore the bifurcations associated with the transition of stability among the equilibria. For the numerical illustration we choose the parameter values in equations \eqref{eq:sis_scalar_com}, \eqref{eq:main_zs}, and \eqref{eq:main_zi} as  summarized below: 

\begin{center}
	\begin{tabular}{|c|c|c|c|c|c|c|}
	\hline
	$c_P$ & 	$\alpha$ & $\beta_u$ &	$c_{\IUt}$ & $L$ &  $\beta_p$ & $c_{\IPt}$  \\ \hline 
	$1$ & $0.5$ & $0.3$ & $2$ & $80$ & $0.15$ & $1$ \\ \hline 
\end{tabular} 
\end{center}

We adopt the recovery rate $\gamma$ as a bifurcation parameter and use the numerical continuation package MATCONT~\cite{dhooge2003matcont} to compute the bifurcation diagram. The bifurcation diagram is shown in Fig.~\ref{fig:bifurcation-diagram}.

\begin{figure}[ht!]
\includegraphics[width=\linewidth]{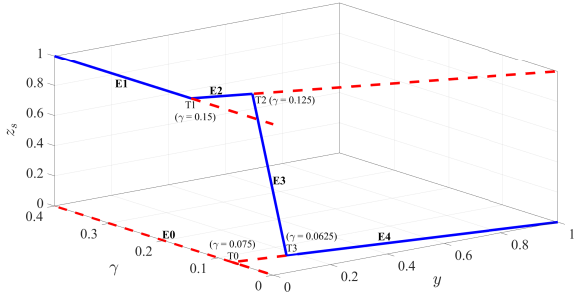} \\
\caption{Bifurcation diagram of the equilibria of the coupled SIS epidemic \\
	and evolutionary behavioral dynamics. (solid blue: stable branch and dashed red: unstable branch of equilibria)} \label{fig:bifurcation-diagram}
\end{figure}

For $\gamma \to 0^+$, $\mathbf{E4}$ is the stable equilibria, while $\mathbf{E0}$, $\mathbf{E1}$, and $\mathbf{E2}$ are unstable. As the value of $\gamma $ is increased at point T3 in  Fig.~\ref{fig:bifurcation-diagram}, $\mathbf{E4}$  exchanges stability to $\mathbf{E3}$ in a transcritical bifurcation. Note that the unstable branch of $\mathbf{E3}$ is not visible since it is associated with negative values of $z_S$. As the value of $\gamma$ is increased, the fraction of susceptible population adopting the protection decreases and at T2, $\mathbf{E3}$  exchanges stability to $\mathbf{E2}$  in another transcritical bifurcation. Again, the unstable branch of $\mathbf{E3}$ at T2 corresponds to $z_s>1$ and is not visible in the bifurcation diagram. Upon further increasing $\gamma$, the fraction of the infected population continues to decrease, and at T1, $\mathbf{E2}$ exchanges stability with the disease-free equilibria $\mathbf{E1}$.  The unstable branch of $\mathbf{E2}$ at T1 corresponds to negative values of $y$ and is not visible. 

Another transcritical bifurcation takes place at T0 $(\gamma= \alpha \beta_p)$, where $\mathbf{E0}$ and $\mathbf{E4}$ cross.   For $\gamma < \alpha \beta_p$ (resp., $\gamma > \alpha \beta_p$), $\mathbf{E0}$ has two (resp., one) eigenvalues in the right-half plane. As $\mathbf{E4}$ approaches T0 from $y >0$, it has one eigenvalue in the right-half plane, while for $y<0$ near T0, $\mathbf{E4}$ is stable. Thus, the transcritical bifurcation at T0 corresponds to the exchange of the stability of stable and unstable eigenvalues of two unstable equilibrium points.  

\begin{figure}[ht!]
\includegraphics[width=\linewidth]{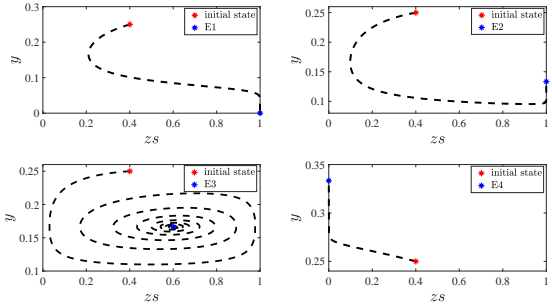} \\
\caption{Trajectories for equilibrium points $\mathbf{E1}$, $\mathbf{E2}$, $\mathbf{E3}$, and $\mathbf{E4}$. The top panels shows trajectory for $\gamma=0.2$  and $\gamma=0.13$ at which $\mathbf{E1}$ and $\mathbf{E2}$  are  stable, respectively. The bottom panels shows trajectory for $\gamma=0.1$  and $\gamma=0.05$ at which $\mathbf{E3}$ and $\mathbf{E4}$  are  stable, respectively. 
\label{fig:equilibria}}
\end{figure}

Fig.~\ref{fig:equilibria} illustrates the behavior of trajectories for equilibrium points $\mathbf{E1}$, $\mathbf{E2}$, $\mathbf{E3}$, and $\mathbf{E4}$. For  $\mathbf{E1}$, $\mathbf{E2}$ and $\mathbf{E4}$, the trajectories evolve such that they first approach $z_S =0$, i.e., every susceptible individual adopts protection and subsequently, they approach $z_S =0$ or $1$ depending on the equilibria. In contrast, trajectories for $\mathbf{E3}$ are highly oscillatory which we investigate further in the next section. 

\section{Coupled  Epidemic-Behavioral Dynamics under Timescale Separation}\label{sec:timescale-separation}

We now study the coupled epidemic-behavioral dynamics \eqref{eq:sis_scalar_com}, \eqref{eq:main_zs}, and \eqref{eq:main_zi} under timescale separation to obtain further insights into their behavior. 

\subsection{Replicator dynamics being faster than epidemic dynamics}

In this regime, the coupled dynamics is modeled as a slow-fast system:
\begin{align} \label{eq:coupled-dynamics-timescale_slow_epi}
\begin{split}
\dot{y}(t) & =  f_y(y(t),z_S(t),z_I(t)) \\
\epsilon \dot{z}_{S}(t) & =  f_S(y(t),z_S(t),z_I(t)) \\
\epsilon \dot{z}_{I}(t) & = f_I(y(t),z_S(t),z_I(t)), 
\end{split}
\end{align}
where $\epsilon \in (0,1]$ is a timescale separation variable \cite{berglund2006noise}. 

At a given epidemic prevalence $y$, we characterize the (stable) equilibria of the fast system involving the replicator dynamics with states $(z_S,z_I)$. For reasons discussed earlier, we focus on equilibria with $z_I = 0$. It is now easy to see that if $y \neq y^*_{\mr{int}}$, there are two equilibrium points: $(0,0)$ and $(1,0)$. If $y = y^*_{\mr{int}}$, then $(z_S,0)$ is an equilibrium point of the fast system for any $z_S \in [0,1]$. Following analogous arguments as the proof of Proposition \ref{prop:equilibria-and-stability}, it follows that  $(0,0)$ is locally stable for the fast system when $y > y^*_{\mr{int}}$ and $(1,0)$ is locally stable for the fast system when $y < y^*_{\mr{int}}$. Consequently, we obtain the following reduced dynamics for the slow system: 
\begin{align}
\dot{y}(t) = 
\begin{cases}
\big[(1-y(t)) \beta_p - \gamma \big] y(t), \quad & \text{if } y < y^*_{\mr{int}}, 
\\ \big[(1-y(t)) \alpha \beta_p - \gamma \big] y(t), \quad & \text{if } y > y^*_{\mr{int}},
\end{cases}
\label{eq:epi_slow}
\end{align}
which approximates the coupled dynamics \eqref{eq:coupled-dynamics-timescale_slow_epi} in the limit $\epsilon \to 0$. We now establish the limiting behavior of the above dynamics \eqref{eq:epi_slow}. 

\begin{proposition}[Trajectories under fast behavioral response]\label{prop:epidemic_slow}
	For the epidemic dynamics \eqref{eq:epi_slow} with $y(0) \neq 0$,  the following statements hold
	\begin{enumerate}
		\item if $ y^*_u < 0$, then $y(t)$ monotonically decreases and converges to the origin; 
		\item if  $0 < y^*_u < y^*_{\mr{int}}$, then $y(t)$ monotonically converges to $y_u^*$; 
		\item if  $ y^*_p < y^*_{\mr{int}} < y^*_u$, then $y(t)$ converges to $ y^*_{\mr{int}}$, and the convergence may not be  monotonic;
		\item if  $y^*_{\mr{int}} < y^*_p$, then $y(t)$ monotonically converges to $y_p^*$.
	\end{enumerate}
\end{proposition}

\ignore{
\begin{proof}
\end{proof}
}

The proof is presented in the Appendix.

\begin{remark}
The dynamics in \eqref{eq:epi_slow} potentially represent a class of non-pharmaceutical interventions where authorities impose social distancing measures that reduce the infection rate by a factor $\alpha$ when the infection prevalence exceeds a threshold. Thus, the result in Proposition \ref{prop:epidemic_slow} is  of potential  independent interest. In prior work \cite{hota2020impacts}, an analogous result was obtained in the special case where $\alpha = 0$. 
\end{remark}

\subsection{Epidemic dynamics being faster than behavior}

When epidemic dynamics evolves faster than the adaptation of protective behavior, we arrive at the dynamics 
\begin{align} \label{eq:coupled-dynamics-timescale_fast_epi}
\begin{split}
\epsilon \dot{y}(t) & =  f_y(y(t),z_S(t),z_I(t)) \\
\dot{z}_{S}(t) & =  f_S(y(t),z_S(t),z_I(t)) \\
\dot{z}_{I}(t) & = f_I(y(t),z_S(t),z_I(t)), 
\end{split}
\end{align}
where $\epsilon \in (0,1]$ as before. Let $\subscr{\beta}{eff}(z_S,z_I) := (z_S + \alpha (1-z_S)) (\beta_u z_I + \beta_p (1-z_I))$ denote the effective infection rate as a function of player strategies. Then, in the limit $\epsilon \to 0$, the infection prevalence converges to the stable equilibrium point of the fast system given by
\begin{equation}\label{eq:11}
y^*(z_S,z_I) = 
\begin{cases}
0, \quad & \text{if } \subscr{\beta}{eff}(z_S,z_I) < \gamma,
\\ 1-\frac{\gamma}{\subscr{\beta}{eff}(z_S,z_I)}, \quad & \text{if } \subscr{\beta}{eff}(z_S,z_I) > \gamma. 
\end{cases}
\end{equation}

Consequently, we obtain the following reduced dynamics for the slow system: 
\begin{align} \label{eq:slow-dynamics-timescale_fast_epi}
\begin{split}
\dot{z}_{S}(t) & =  f_S(y^*(z_S(t),z_I(t)),z_S(t),z_I(t)) \\
\dot{z}_{I}(t) & = f_I(y^*(z_S(t),z_I(t)),z_S(t),z_I(t)), 
\end{split}
\end{align}
which approximates the behavioral dynamics when $\epsilon \to 0$.

\begin{remark}
The slow timescale dynamics~\eqref{eq:slow-dynamics-timescale_fast_epi} sets the value of $\subscr{\beta}{eff}(z_S(t),z_I(t))$ for the faster $y$ dynamics in~\eqref{eq:coupled-dynamics-timescale_fast_epi}. The class of  problems in which the bifurcation parameter of a differential equation is varied on a slower timescale are referred to as dynamic bifurcation problems~\cite{berglund2006noise}. It follows from the theory of dynamic bifurcations that when $\subscr{\beta}{eff}(z_S(t),z_I(t))$  is not in the neighborhood of the bifurcation point $\subscr{\beta}{eff}(z_S(t),z_I(t))  = \gamma$, the system trajectories track the stable branches of the bifurcation diagram. System trajectories may not track the stable branch of the bifurcation diagram in a neighborhood of  bifurcation point, which leads to the so-called \emph{bifurcation delay}~\cite{berglund2006noise}.  
\end{remark}

\subsection{Numerical Results}

We now provide further insights into the dynamics under timescale separation, which exhibits large oscillations and the bifurcation delay phenomenon via numerical simulations. 

We first  investigate the behavior of  system~\eqref{eq:coupled-dynamics-timescale_slow_epi} for $\epsilon \in \{0.01, 0.1, 1\}$.  The same model parameters as in Section~\ref{sec:bif-analysis} are selected and $\gamma$ is selected as $0.1$ so that $\mathbf{E3}$ is the stable equilibrium point. Fig.~\ref{fig:slow_fast_ga_p1} (right panel) shows the time-evolution of $y$ and $z_S$. More oscillatory behavior is observed as $\epsilon$ becomes smaller, i.e., the behavioral dynamics becomes faster. To understand this behavior, we focus on $z_S$ dynamics~\eqref{eq:main_zs} with 
	\[\Delta F =   c_P - L(1-\alpha)(\beta_u z_{I}(t)  + \beta_p (1-z_I(t))) y(t)
\]
as a dynamic parameter. To this end, we  illustrate $z_S$ trajectories in $z_S-\Delta F$ plane in  Fig.~\ref{fig:slow_fast_ga_p1} (left panel).  

Recall that if $\Delta F$ is a positive (resp. negative) constant, then $z_S =1$ (resp. $z_S=0$) is a stable equilibrium point. Accordingly, $z_S  =0$ and $z_S=1$ are marked blue and red in Fig.~\ref{fig:slow_fast_ga_p1} (left panel), when they are stable and unstable, respectively. 
	
Since behavioral dynamics is fast, $y$ is quasi-stationary and $z_I$ very quickly converges to zero. In Fig.~\ref{fig:slow_fast_ga_p1} (left panel), the initial fraction of the infected population is sufficiently high such that $\Delta F <0$, then the fast behavioral dynamics quickly converge to $z_S =0$ (the bottom solid blue line), i.e., every susceptible individual adopts protection. This results in a decrease in the fraction of the infected population and increases $\Delta F$. As $\Delta F$ becomes positive, $z_S=0$ becomes unstable and $z_S$ quickly jumps to $z_S =1$ (the top solid blue line), and a similar process repeats which again drives $\Delta F$ to negative values. This process leads to the highly oscillatory behavior seen in Fig.~\ref{fig:slow_fast_ga_p1}. Eventually, trajectories converge such that $\Delta F =0$, and in this regime $z_S$ settles to an equilibrium value in the interval $(0,1)$.

\begin{figure}[ht!]
	\centering
	\includegraphics[width=\linewidth]{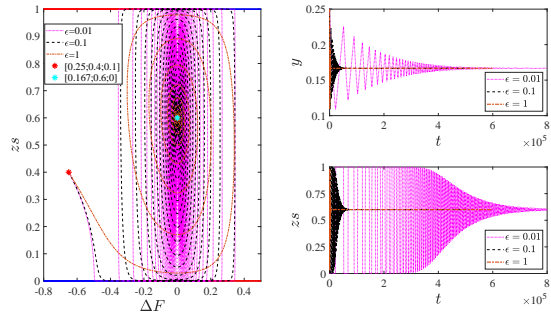}
	\caption{Trajectories of the coupled epidemic-behavioral dynamics when the behavioral dynamics evolve on a faster timescale. The left panel shows trajectories in $z_S-\Delta F$ plane. For $\Delta F <0$ (resp. $\Delta F >0$) $z_S =0$ (resp. $z_S =1$) is the stable equilibrium point and is shown in solid blue line. The slower epidemic dynamics drive $\Delta F$ to change sign leading to an abrupt switch in the equilibrium value of $z_S$, which in turn causes $\Delta F$ to again switch sign. This process repeats and results in oscillatory trajectories. Finally, the trajectories reach an equilibrium point at which $\Delta F =0$. The right panel shows the oscillatory trajectories. The oscillations increase as $\epsilon$ is decreased.}
	\label{fig:slow_fast_ga_p1}
\end{figure}

We now investigate the behavior of  system~\eqref{eq:coupled-dynamics-timescale_fast_epi}   for $\epsilon \in \{0.0001, 0.01, 1\}$. The same model parameters as above are selected. The time evolution of $y$ and $z_S$ trajectories is shown in Fig.~\ref{fig:fast_slow} (right panels). To better understand these trajectories,  we investigate them  in the $y-\subscr{\beta}{eff}$ plane in Fig.~\ref{fig:fast_slow}, left panel), where $ \subscr{\beta}{eff}(z_S,z_I) = (z_S + \alpha (1-z_S)) (\beta_u z_I + \beta_p (1-z_I))$ is a function of $z_I$ and $z_S$. The solid blue and red lines show the bifurcation diagram of the faster SIS epidemic dynamics with $\subscr{\beta}{eff}$ as the bifurcation parameter.  
	
In Fig.~\ref{fig:fast_slow_ga_p1}, the initial condition is selected such that  the initial value of $\subscr{\beta}{eff}$ corresponds to endemic equilibria being stable. For small $\epsilon$ values, the trajectories very quickly converge to the stable branch of the bifurcation diagram and traverse along with the bifurcation diagram until they reach the stable equilibria $\mathbf{E3}$ of the coupled dynamics. This behavior is consistent with the theory of dynamic bifurcations~\cite{berglund2006noise}. For moderate values of $\epsilon$, similar behavior is observed but trajectories overshoot the stable branch and subsequently return to it. For $\epsilon=1$, we observe oscillatory behavior similar to the case in which the behavioral dynamics are faster.

In Fig.~\ref{fig:fast_slow_ga_p07},  the initial value of $\subscr{\beta}{eff}$ corresponds to disease-free equilibria being stable. For small $\epsilon$, the trajectories converge very quickly to the stable branch of the bifurcation diagram for the SIS epidemic dynamics and track the stable branch until they reach the bifurcation point. Subsequently, they continue on the unstable branch for some time before switching to the stable branch. This behavior is referred to as \emph{bifurcation delay} \cite{berglund2006noise} and is attributed to the fact that at the bifurcation point, the SIS epidemic dynamics become slower than the otherwise slower behavioral dynamics.  Similar behavior is observed for moderate $\epsilon$, while for $\epsilon=1$ we observe oscillatory behavior as earlier. 

 \begin{figure}[ht!]
 	\begin{subfigure}[h]{\linewidth}
 				\centering
 				\includegraphics[scale=0.22]{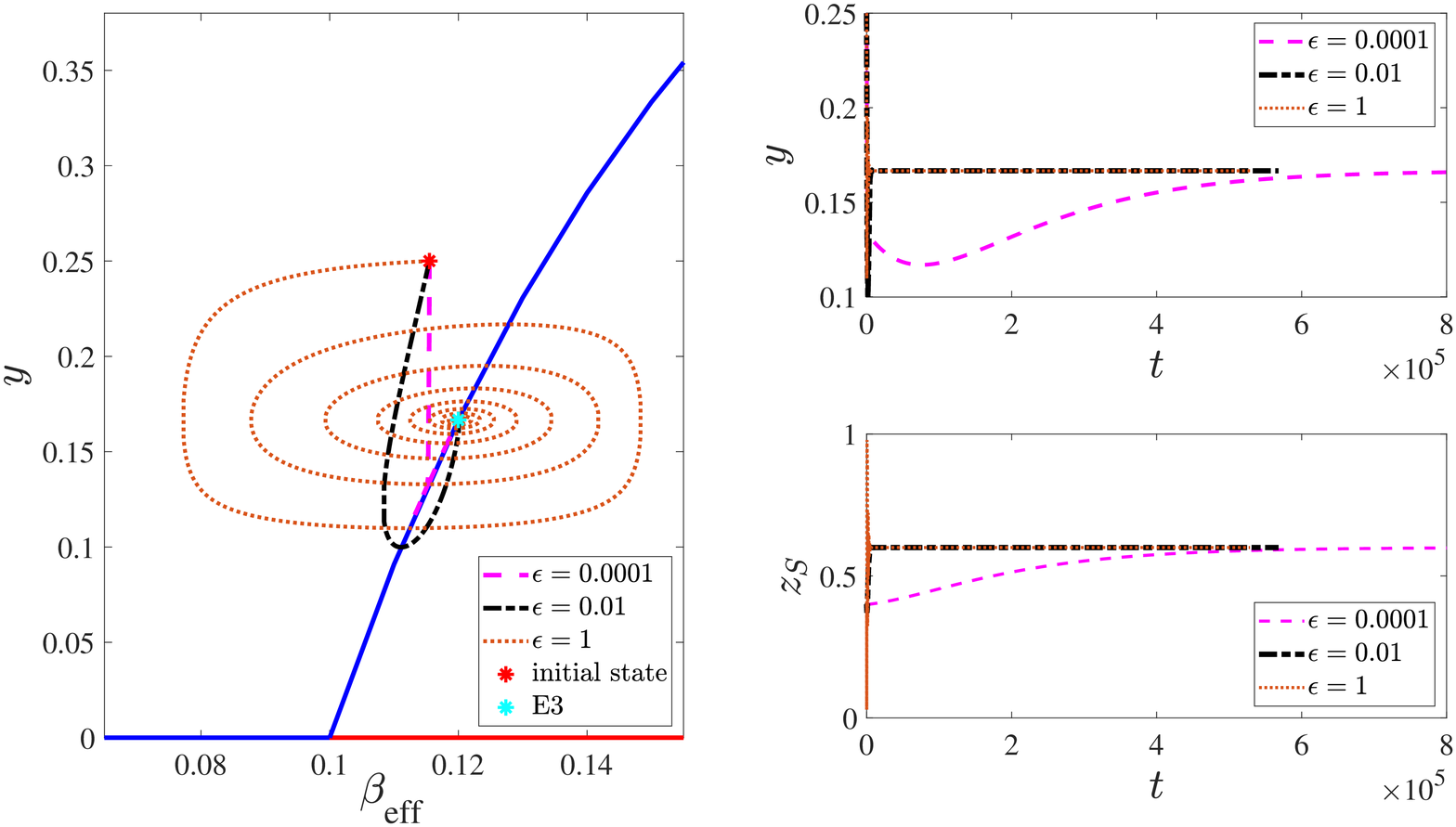}
 				\caption{}
 				\label{fig:fast_slow_ga_p1}
 	\end{subfigure}
 	\begin{subfigure}[h]{\linewidth}
	           \centering
               \includegraphics[scale=0.22]{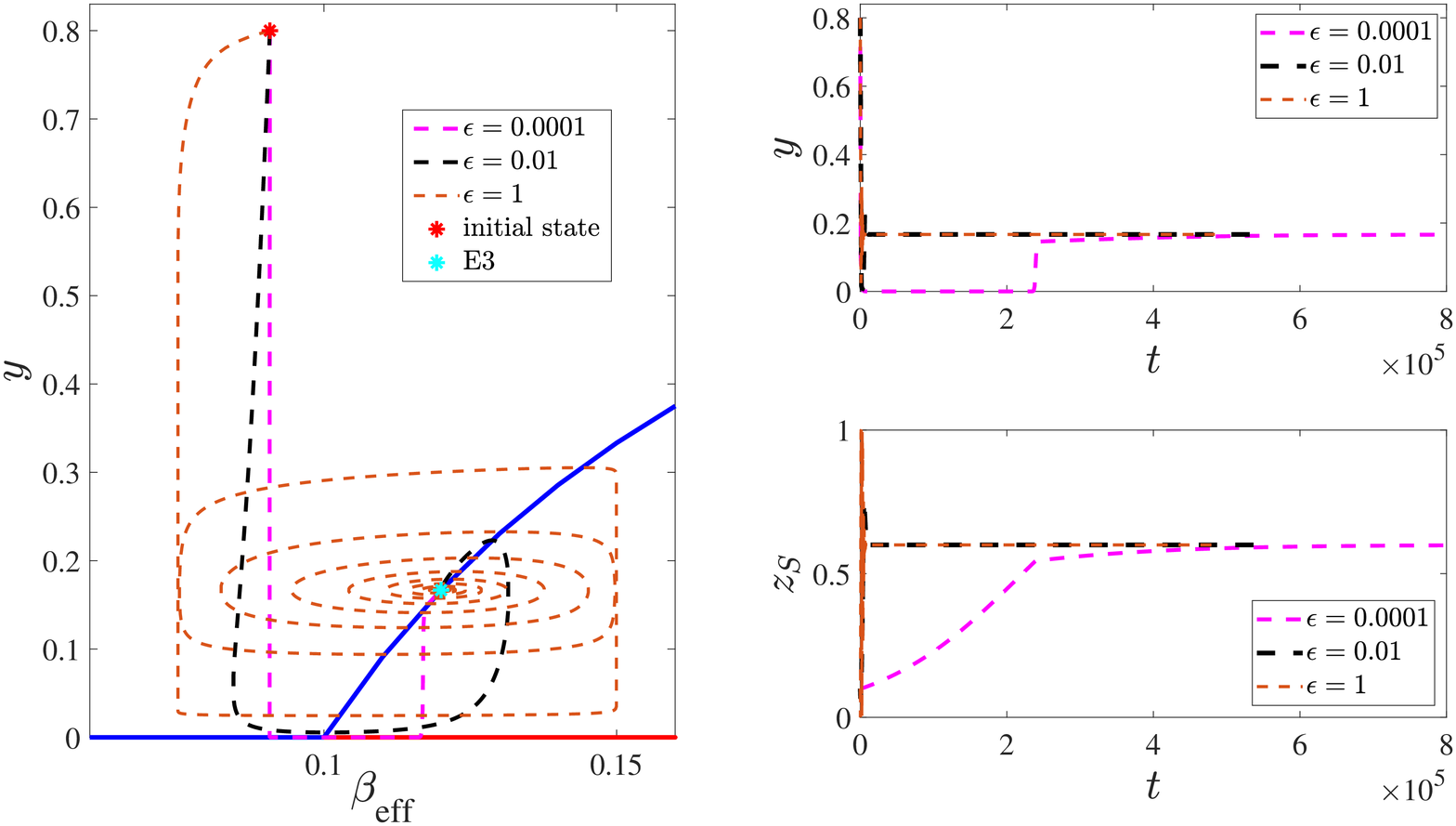}
               \caption{}
               \label{fig:fast_slow_ga_p07}
 	\end{subfigure}
\caption{Trajectories of the coupled epidemic-behavioral dynamics when the epidemic dynamics evolve on a faster timescale. The left panels shows trajectories in $y- \subscr{\beta}{eff}$ plane. The right panels show the time-evolution of $y$ and $z_S$ trajectories. Left panel (a) shows for small $\epsilon$, trajectories quickly converge to the stable branch of the bifurcation diagram and then move along the bifurcation diagram to reach the stable equilibria.  Left panel (b) shows for small $\epsilon$, trajectories quickly converge to the stable branch of the bifurcation diagram and then move along on the bifurcation diagram until they reach the bifurcation point. In a neighborhood of the bifurcation point, trajectories do not move along the stable branch of the bifurcation diagram but they return to the stable branch after  some delay. Oscillatory behavior is observed for larger values of $\epsilon$. The right panel shows the time evolution of $y$ and $z_S$.
\label{fig:fast_slow}}

\end{figure}

\section{Conclusions}\label{sec:conclusions}

We proposed and analyzed a novel model that captures the interaction of epidemic propagation dynamics with human protection adoption behavioral dynamics. For the proposed coupled epidemic-behavioral dynamics, we completely characterized the equilibrium points, their stability properties, and the associated bifurcations. The set of equilibria corresponds to a disease-free population or an endemic, and complete, partial, or no adoption of protection by susceptible population. We showed that the stability of these equilibria is a function of infection and recovery rate as well as the utility functions of the  individuals in the population. We investigated the transient behavior of the model under two timescale scenarios: (i) faster behavioral dynamics, (ii) faster epidemic dynamics. In the former case, we showed that the coupled dynamics may lead to highly oscillatory behavior and may take a long time to converge to the stable equilibria. In the latter case, we illustrated the bifurcation delay phenomenon in the onset of endemic. 

\bibliographystyle{ieeetran}
\bibliography{refs_new}


\appendix

\textbf{Proof of Proposition \ref{prop:equilibria-and-stability}:} It can be verified that $\mathbf{E0}$ and $\mathbf{E1}$  are always equilibria of the coupled dynamics. The Jacobian matrix at $\mathbf{E0}$ and $\mathbf{E1}$ are 
\begin{equation*}
	J_{\mathbf{E0}} = 
	\left[\begin{smallmatrix}
		\alpha \beta_p - \gamma & 0 & 0\\
		0 & c_P & 0\\
		0 & 0 & c_{\IPt} - c_{\IUt}
	\end{smallmatrix}\right], \text{ and } 	J_{\mathbf{E1}} = 
\left[\begin{smallmatrix}
\beta_p - \gamma & 0 & 0\\
0 & -c_P & 0\\
0 & 0 & c_{\IPt} - c_{\IUt}
\end{smallmatrix}\right]. 
\end{equation*}
 $J_{\mathbf{E0}}$ has a positive eigenvalue if $c_P > 0$. Thus, for any nonzero cost of adopting protection, $\mathbf{E0}$ is not a stable equilibrium point of the coupled dynamics. Likewise, $\mathbf{E1}$ is stable if $\beta_p < \gamma$, and unstable, otherwise. 

We now analyze equilibrium points where infection is endemic. It can be verified that $\mathbf{E2}$  exists only when $\beta_p > \gamma$. The Jacobian matrix at $\mathbf{E2}$ is
\begin{equation*}
	J_{\mathbf{E2}} = 
	\left[\begin{smallmatrix}
		d_1 & d_2 & d_3\\
		0 & -[c_P - L(1-\alpha)\beta_py^*_u] & 0\\
		0 & 0 & c_{\IPt} - c_{\IUt}
	\end{smallmatrix}\right]
	,
\end{equation*}
where $d_1 = (1-2y^*_u)\beta_p - \gamma$, $d_2 = y^*_u(1-y^*_u)(1-\alpha)\beta_p > 0$ and $d_3 = y^*_u(1-y^*_u)(\beta_u - \beta_p) > 0$. Thus, $J_{\mathbf{E2}}$ is an upper triangular matrix. Furthermore, the first diagonal entry is
\begin{align*}
	(1-2y^*_u)\beta_p - \gamma & 
	= \gamma - \beta_p < 0,
\end{align*}
in the regime where $\mathbf{E2}$ exists. Therefore, $\mathbf{E2}$ is stable when 
\begin{align*}
	c_P & > L(1-\alpha)\beta_p y^*_u  \iff  y^*_u  < y^*_{\mr{int}}.
\end{align*}

It can be verified that $\mathbf{E3} = (y^*_{\mr{int}},z^*_{S,\mr{int}},0)$ is an equilibrium point. We now derive conditions that ensure $y^*_{\mr{int}},z^*_{S, \mr{int}} \in [0,1]$. We now examine the conditions under which $y^*_{\mr{int}} \in (0,1)$ and $z^*_{S,\mr{int}} \in (0,1)$.  By definition, $y^*_{\mr{int}} > 0$. We now observe that
\begin{align*}
	z^*_{S,\mr{int}} > 0 & \iff \frac{\gamma}{\alpha \beta_p} > 1 - y^*_{\mr{int}} \iff  y^*_{\mr{int}} > 1 - \frac{\gamma}{\alpha \beta_p} \\
	\text{and }	z^*_{S,\mr{int}} < 1 & \iff\frac{\gamma}{\beta_p} < 1 - y^*_{\mr{int}} \iff  y^*_{\mr{int}} < 1 - \frac{\gamma}{\beta_p}. 
\end{align*}
Thus,  $\mathbf{E3}$ exists when $ y^*_p< y^*_{\mr{int}} <  y^*_u$.  Note that the third row of the Jacobian of the dynamics at $\mathbf{E3}$,  $J_{\mathbf{E3}}$, would be $[0 \quad 0 \quad c_{\IPt} - c_{\IUt}]$ as before, and as a result, $c_{\IPt} - c_{\IUt} < 0$ would be an eigenvalue. Thus, we focus on the $2 \times 2$ sub-matrix containing the first two rows and columns of the Jacobian matrix which simplifies to
\begin{align*}
	\hat{J}_{\mathbf{E3}} = 
	\left[\begin{smallmatrix}
		\frac{-\gamma y^*_{\mr{int}}}{1-y^*_{\mr{int}}} & d_4\\
		- z^*_{S,\mr{int}}(1-z^*_{S,\mr{int}})L(1-\alpha)\beta_p & 0\\
	\end{smallmatrix}\right],
\end{align*}
where $d_4 = y^*_{\mr{int}}(1-y^*_{\mr{int}})(1-\alpha)\beta_p$. For the above matrix, the sum of the eigenvalues is negative and the determinant is positive, and as a result, all three eigenvalues of $J_{\mathbf{E3}}$ are negative. Therefore, $\mathbf{E3}$, when it exists, is a stable equilibrium of the coupled dynamics. 

It can be verified that  $\mathbf{E4}$ exists when $y^*_p \in (0,1)$ or equivalently, when $\gamma < \alpha \beta_p$. The Jacobian matrix at $\mathbf{E4}$ is
\begin{equation}
	J_{\mathbf{E4}} = 
	\left[\begin{smallmatrix}
		(1-2y^*_p)\alpha\beta_p - \gamma & d_5 & d_6\\
		0 & d_7 & 0\\
		0 & 0 & c_{\IPt} - c_{\IUt}
	\end{smallmatrix}\right],
\end{equation}
where $d_5 =  y^*_p(1-y^*_p) (1-\alpha)\beta_p> 0$, $d_6 = y^*_p(1-y^*_p)\alpha(\beta_u - \beta_p)> 0$ and $d_7 = c_P - L(1-\alpha)\beta_py^*_p$. Thus, $J_{\mathbf{E4}}$ is a diagonal matrix with the first diagonal entry
\begin{align*}
	(1-2y^*_p)\alpha\beta_p - \gamma  
	& = \gamma - \alpha\beta_p < 0,
\end{align*}
in the regime where $\mathbf{E4}$ exists. Therefore, $\mathbf{E4}$ is stable when 
\begin{align*}
	c_P & < L(1-\alpha)\beta_p y^*_p  \iff   y^*_p  > y^*_{\mr{int}}.
\end{align*}
This concludes the proof. 

\ignore{We evaluate the entries of the Jacobian matrix of the coupled dynamics \eqref{eq:sis_scalar_com}, \eqref{eq:main_zs} and \eqref{eq:main_zi} required to determine the stability of the equilibrium points. Specifically, we compute
\begin{align*}
\frac{\partial f_y}{\partial y} & = \big[(1-y) (z_S + \alpha (1-z_S)) (\beta_u z_I + \beta_p (1-z_I)) - \gamma \big] 
\\ & \qquad - y (z_S + \alpha(1-z_S)) (\beta_u z_I + \beta_p (1-z_I)), 
\\ \frac{\partial f_y}{\partial z_S} & = y \big[(1-y) (1-\alpha) (\beta_u z_I + \beta_p (1-z_I)) \big],
\\ \frac{\partial f_y}{\partial z_I} & = y \big[(1-y) (z_S + \alpha (1-z_S)) (\beta_u - \beta_p)\big].
\end{align*}
Similarly, 
\begin{align*}
\frac{\partial f_S}{\partial y} & = - z_{S}(1-{z}_{S}) L(1-\alpha)(\beta_u z_{I} + \beta_p (1-z_I)),
\\ \frac{\partial f_S}{\partial z_S} & = (1-2z_S) \big[c_P - L(1-\alpha)(\beta_u z_{I} + \beta_p (1-z_I)) y\big]
\\ \frac{\partial f_S}{\partial z_I} & = - z_{S}(1-{z}_{S}) L(1-\alpha)(\beta_u - \beta_p) y
\end{align*}
Finally,
\begin{align*}
\frac{\partial f_I}{\partial y} & = 0, \quad  \frac{\partial f_I}{\partial z_S}  = 0, \quad \frac{\partial f_I}{\partial z_I}  = (1-2z_I) (c_{\IPt}-c_{\IUt}).
\end{align*}}

\textbf{Proof of Proposition \ref{prop:epidemic_slow}:} Recall that since $\alpha \in (0,1)$, we have $y^*_p < y^*_u$. It follows from \cite{mei2017dynamics} that the dynamics 
$$ \dot{y}(t) = [(1-y(t)) \beta - \gamma] y(t) $$
converges to $0$ if $\gamma \geq \beta$ and converges to $y^* = 1-\frac{\gamma}{\beta}$, otherwise. Furthermore, it is easy to see that in case of the latter, $|y(t) - y^*|$ is monotonically decreasing.\footnote{For $V(y) = (y-y^*)^2$, we have $\dot{V}(y) = -2\beta yV(y)$.}  Note that when $y < y^*_{\mr{int}}$ (respectively, $y > y^*_{\mr{int}}$), \eqref{eq:epi_slow} is analogous to the above dynamics with $\beta = \beta_p$ (respectively, $\beta=\alpha \beta_p$). We now analyze the four cases stated above. 

\noindent {\bf Case 1: $y^*_u < 0$}. In this case, we have $\alpha \beta_p < \beta_p < \gamma$ and as a result, $y = 0$ is the only equilibrium both when $y > y^*_{\mr{int}}$ and $y < y^*_{\mr{int}}$. 

\noindent {\bf Case 2: $0 < y^*_u < y^*_{\mr{int}}$}. Since $y^*_p < y^*_u$, $y(t)$ is monotonically decreasing when $y(t) > y^*_{\mr{int}}$. Similarly, $y(t)$ is monotonically decreasing when $y_u^* < y(t) < y^*_{\mr{int}}$ and monotonically increasing for $y(t) < y_u^*$.
Thus, the claim follows. 

\noindent {\bf Case 3: $y^*_p < y^*_{\mr{int}} < y^*_u$}. In this case, $y(t)$ is monotonically decreasing when $y(t) > y^*_{\mr{int}}$ as $\max(0,y^*_p) < y^*_{\mr{int}}$, and monotonically increasing when $y(t) < y^*_{\mr{int}}$ since the corresponding stable equilibrium $y^*_u > y^*_{\mr{int}}$. Thus, $y^*_{\mr{int}}$ acts as a sliding surface for the dynamics \eqref{eq:epi_slow}. $y(t)$ may switch between two subcases in~\eqref{eq:epi_slow} leading to an oscillatory behavior as will be illustrated numerically. 

\noindent {\bf Case 4: $y^*_p > y^*_{\mr{int}}$}. In this case, $y(t)$ is monotonically increasing when $y(t) < y^*_{\mr{int}}$ since $y^*_u > y^*_p > y^*_{\mr{int}}$. Monotonic behavior for $y(t) > y^*_{\mr{int}}$ follows from analogous arguments as Case 2. 

This concludes the proof.

\end{document}